\documentclass[11pt,a4paper]{article}

\usepackage[T1]{fontenc}
\usepackage[utf8]{inputenc}
\usepackage{lmodern}
\usepackage[english]{babel}
\usepackage[a4paper,margin=2.6cm]{geometry}
\usepackage{amsmath,amssymb,amsthm,bm}
\usepackage{graphicx}
\usepackage{enumitem}
\usepackage[font=small,labelfont=bf,labelsep=period,singlelinecheck=false]{caption}
\usepackage[numbers,sort&compress]{natbib}
\usepackage{microtype}
\usepackage[hidelinks]{hyperref}

\hypersetup{
  pdftitle={A theory of plasticity: capacity for change as inverse configurational constraint},
  pdfauthor={Igor Branchi},
  pdfsubject={Plasticity, configurational constraint, barrier spectrum, network theory},
}

\setlist{itemsep=2pt,topsep=4pt}
\theoremstyle{plain}
\newtheorem{lemma}{Lemma}
\newtheorem{theorem}{Theorem}
\newtheorem{proposition}{Proposition}
\newtheorem{corollary}{Corollary}
\theoremstyle{definition}
\newtheorem*{definition}{Definition}
\newtheorem{remark}{Remark}

\title{\textbf{A theory of plasticity:}\\[2pt] capacity for change as inverse configurational constraint}

\author{Igor Branchi\thanks{Correspondence: \href{mailto:igor.branchi@iss.it}{igor.branchi@iss.it}. ORCID: \href{https://orcid.org/0000-0003-4484-3598}{0000-0003-4484-3598}.}\\[4pt]
\normalsize Center for Behavioral Sciences and Mental Health,\\
\normalsize Istituto Superiore di Sanit\`a, Viale Regina Elena 299, 00161 Rome, Italy}

\date{November 2026\\[2pt] \small First version: September 2026}

\begin{document}

\maketitle

\begin{abstract}
Plasticity is invoked across the sciences, from cell-state transitions to neural and behavioral adaptation, to explain how systems can change, yet it is often inferred from the very change it is meant to explain. A system may have many alternatives, realize none and still be plastic. Another may be driven far toward its only alternative, but the magnitude of that change does not establish its plasticity. What matters for plasticity is not how far the system moves but how strongly its present configuration constrains alternatives.

Here I propose that plasticity, understood as a prospective property, is inverse configurational constraint: a system is plastic to the extent that its present configuration weakly constrains displacement toward alternatives within a declared representation. For systems sharing a functional partition, structural representation, and normalized architecture, increasing plasticity lowers every positive finite structural barrier separating the current state from its alternatives. With a fixed accessibility criterion, the accessible repertoire can expand but cannot shrink. When normalized architectures differ, barrier ordering can still be preserved if the increase in plasticity is sufficiently large relative to architectural differences.

In network representations, coupling provides a coefficient-level realization of configurational constraint, whose inverse defines plasticity. Barrier reduction and repertoire nesting are consequences, not components, of the definition. This connects two research traditions: quantifying connectivity without interpreting it as plasticity, and treating plasticity as a capacity for change without a prospective operational measure. The framework provides a common theoretical language for understanding plasticity across levels of biological organization and beyond biology.

\end{abstract}

\section{Introduction}\label{sec:1}

Plasticity is among the most widely invoked and inconsistently defined concepts across the sciences (\protect\hyperlink{_ENREF_31}{Seblani et al., 2024}). In the biological sciences, it is inferred from changes in gene expression, cell state, synaptic strength, long-term potentiation and depression, critical-period reopening, changes in functional connectivity, and behavioral adaptation (\protect\hyperlink{_ENREF_2}{Bliss and Lomo, 1973}; \protect\hyperlink{_ENREF_9}{Citri and Malenka, 2008}; \protect\hyperlink{_ENREF_24}{Hensch, 2005}; \protect\hyperlink{_ENREF_40}{West-Eberhard, 2003}). These measures differ in scale and mechanism. More importantly, they quantify manifestations of change without specifying the common theoretical property that makes them measures of plasticity. They thereby reveal a recurrent ambiguity between capacity for change and realized change that raises a fundamental problem in theoretical biology. Biological systems can retain the capacity to undergo functional transitions without expressing that capacity under the conditions observed, as in developmental landscapes and critical periods (\protect\hyperlink{_ENREF_24}{Hensch, 2005}; \protect\hyperlink{_ENREF_36}{Waddington, 1957}). Conversely, a large response to a perturbation does not necessarily imply that many alternative functional configurations are accessible. The problem also has practical consequences. In development, learning, and mental health, the relevant question is often whether a system has the capacity to change before change is observed. Whether a nervous system can still be reshaped by experience, or whether a patient will respond to treatment or to a change of environment, must be anticipated rather than observed after the fact (\protect\hyperlink{_ENREF_5}{Branchi, 2022}; \protect\hyperlink{_ENREF_6}{Branchi, 2023}).

The difficulty is not merely terminological. Most operational definitions infer plasticity from a difference between measurements, a rate of modification, or a response following an intervention. Such definitions require the system to change before its plasticity can be estimated (\protect\hyperlink{_ENREF_5}{Branchi, 2022}; \protect\hyperlink{_ENREF_6}{Branchi, 2023}).

Yet plasticity is a prospective property. A system with many available alternatives may realize none of them and still be plastic. Conversely, a system may be driven a long way toward its only available alternative, but the magnitude of that change does not establish its plasticity. Change is therefore neither necessary nor sufficient. What matters is how strongly the present configuration constrains the alternatives available to the system, not the magnitude of the trajectory subsequently realized (\protect\hyperlink{_ENREF_7}{Branchi, 2026}).

This distinction suggests that plasticity should be defined from the present organization of the system rather than from its subsequent trajectory. The proposal developed here is:
\[\boxed{\text{Plasticity is inverse configurational constraint}}\]
A system is plastic to the extent that its present configuration weakly constrains displacement toward alternatives. It is rigid to the extent that its present configuration strongly constrains such displacement.

The definition is structural and representation-relative. It concerns how strongly the current configuration constrains access to alternatives. This constraint can influence transition rates but does not by itself determine those rates, the alternative realized, or the ensuing trajectory.

Because any scalar measure of constraint depends on how the relevant structural contributions are represented, comparisons require a structural representation and normalization convention declared in advance for the comparison class. Within that representation, inverse constraint provides a prospective quantity whose configurational consequences can be derived.

The identification of inverse configurational constraint with plasticity is a conceptual definition. Its relation to configurational accessibility is a formal consequence. Under the matching conditions specified below, including a common structural representation and normalized architecture, increasing plasticity lowers every positive finite structural barrier by the same factor. In a network representation, the constraint magnitude \(C\) (Section~\ref{sec:4_2}) corresponds to aggregate coupling up to the declared normalization (\protect\hyperlink{_ENREF_3}{Borsboom, 2017}; \protect\hyperlink{_ENREF_17}{Epskamp et al., 2018}). This coupling comprises the self-coupling of each element and the cross-coupling between elements, while connectivity has usually been summarized as cross-coupling alone. As discussed in Section~\ref{sec:9_2}, this coupling can exceed the nominal coupling of direct interactions when unrepresented common causes contribute, and only part of it may be jointly realizable when couplings are incompatible.

The present framework thereby connects two largely separate traditions: one that quantifies connectivity without interpreting it as plasticity, and another that invokes plasticity as a capacity for change without specifying a prospective operational measure. Aggregate coupling represents the constraint imposed by the current configuration while its inverse defines plasticity. Positive homogeneity determines how plasticity orders the local barrier spectrum. A structural accessibility criterion then translates this barrier ordering into an ordering of accessible repertoires. The framework therefore makes it possible to estimate how plastic a system is before it changes and to derive how its accessible repertoire would change if its present constraint were altered.

\section{Plasticity as inverse configurational constraint}\label{sec:2}

\subsection{Conceptual definition}\label{sec:2_1}

\begin{definition}
Plasticity is the inverse of the magnitude of the configurational constraint imposed by a system's present state on access to functionally distinct alternatives, relative to a declared structural representation (\protect\hyperlink{_ENREF_7}{Branchi, 2026}).
\end{definition}

This definition recovers an ordinary use of plasticity that definitions based on realized change cannot capture: a system can be plastic even when it is not currently changing. A nervous system during a critical period, a cell before differentiation, and a patient early in treatment can all be described as plastic even in the absence of ongoing change (\protect\hyperlink{_ENREF_4}{Branchi, 2011}; \protect\hyperlink{_ENREF_24}{Hensch, 2005}; \protect\hyperlink{_ENREF_36}{Waddington, 1957}). Conversely, a system may be driven far toward its only available alternative, but the magnitude of that change does not establish its plasticity.

The definition has four properties. It is \emph{structural}, because it concerns the present organization of the system. It is \emph{prospective}, because it can be estimated before the subsequent trajectory is known. It is \emph{counterfactual}, because it concerns what the system could do from its present configuration rather than what it eventually does. It is \emph{valence-neutral}, because it does not specify whether an ensuing change is beneficial, detrimental, adaptive, or maladaptive.

\subsection{Structure before dynamics}\label{sec:2_2}

The object of the definition is the system as represented at the chosen level of analysis. Its structural representation may include molecular, cellular, neural, physiological, behavioral, or relational contributions.

Locality refers to the basin currently occupied by the system in configurational space, not to an individual component or network node. The theory characterizes the structural constraint on departures from that basin. The configurational space may be continuous or discrete.

Throughout, a configuration \(x\) is a point of configurational space, not a time-dependent trajectory \(x(t)\). Admissible paths used below are geometric routes through configurational space, not equations of motion. Temporal evolution introduces additional quantities, including transition rates, friction, stochastic forcing, kinetic prefactors, and observation horizons (\protect\hyperlink{_ENREF_23}{Hänggi et al., 1990}; \protect\hyperlink{_ENREF_26}{Kramers, 1940}). None enters the definition of plasticity developed here.

\section{Configurational alternatives}\label{sec:3}

Consider a system occupying a local configurational basin \(a\). Let
\[\mathcal{X}\]
denote its configurational space, and let
\[\varphi:\mathcal{X} \rightarrow \mathcal{Y}\]
map physical configurations to functionally relevant outcomes.

Two configurations are treated as realizations of the same functional alternative when
\[x \sim x^{'}\  \Leftrightarrow \ \varphi(x) = \varphi\left( x^{'} \right)\]
Functional alternatives are therefore equivalence classes induced by \(\varphi\), not arbitrary microstates. Two configurations that realize the same functional outcome count as one alternative.

Let
\[\mathcal{Y}_{a}\subseteq \mathcal{Y}\]
denote the set of functionally distinct local alternatives relative to the present basin. For each \(\mu \in \mathcal{Y}_{a}\), let \(\Delta_{a\mu}\) denote the minimum structural barrier from basin \(a\) to any configuration that realizes \(\mu\).

More precisely, let \(x_{a}^{\star}\) denote the minimum of the current basin and let \(\Gamma_{a\mu}\) be the set of admissible structural paths starting at \(x_{a}^{\star}\) and terminating in \(\varphi^{- 1}(\mu)\).

For a continuous configurational space, an admissible path is a continuous map \(\gamma:\lbrack 0,1\rbrack \rightarrow \mathcal{X}\) satisfying
\[\gamma(0) = x_{a}^{\star},\ \ \gamma(1) \in \varphi^{- 1}(\mu)\]
For a discrete configurational space, an admissible path is a finite sequence \(\gamma = \left( x_{0},\ldots,x_{L} \right)\) with \(x_{0} = x_{a}^{\star}\), \(x_{L} \in \varphi^{- 1}(\mu)\), and each consecutive pair \(x_{\ell},x_{\ell +1}\), for \(\ell \in \{0,\ldots,L - 1\}\), differing by an elementary move declared by the model. In a network realization, for example, an elementary move may consist of changing a single network element.

Let \(I_{\gamma}\) denote the corresponding path index set, with \(I_{\gamma} = \lbrack 0,1\rbrack\) in the continuous case and \(I_{\gamma} = \{ 0,\ldots,L\}\) in the discrete case. Let \(V:\mathcal{X}\mathbb{\rightarrow R}\) be a static structural cost landscape. \(V\) assigns to each configuration its structural cost within the declared representation. It need not be a thermodynamic energy, even when it takes the same mathematical form as one, and it does not describe how the system moves through time. Its explicit representation is given in Section~\ref{sec:4_1}. Define
\[\Delta_{a\mu} = \inf_{\gamma \in \Gamma_{a\mu}}\ \sup_{s \in I_{\gamma}}\left( V\left( \gamma(s) \right) - V\left( x_{a}^{\star} \right) \right)\]
Here, a barrier is a minimax difference in structural cost. It acquires a kinetic interpretation as an activation barrier only when an additional dynamical law links structural cost to transitions. If no admissible path connects the current minimum to alternative \(\mu\), we adopt the convention \(\inf\varnothing = \infty\), so that \(\Delta_{a\mu} = \infty\).

Because every admissible path begins at \(x_{a}^{\star}\), \(\Delta_{a\mu} \geq 0\). If several paths or several physical configurations realize the same functional alternative, the infimum selects the least costly route to that alternative. For a smooth landscape satisfying standard mountain-pass conditions, when the infimum is attained, the minimax value corresponds to the relevant separating saddle (\protect\hyperlink{_ENREF_30}{Onuchic et al., 1997}; \protect\hyperlink{_ENREF_38}{Wales, 2003}).

This construction prevents degeneracy from being counted as a multiplicity of functionally distinct alternatives without declaring degeneracy irrelevant. Multiple physical realizations can still lower \(\Delta_{a\mu}\) by providing additional paths to the same outcome. Degeneracy can therefore affect accessibility even though it does not determine the number of functional alternatives (\protect\hyperlink{_ENREF_16}{Edelman and Gally, 2001}; \protect\hyperlink{_ENREF_32}{Tononi et al., 1999}).

The functional map \(\varphi\) is part of the problem specification. Barrier spectra and repertoire sizes can be compared directly only when the systems use the same functional partition.

\section{Structural representation and local constraint}\label{sec:4}

\subsection{Declared and normalized structural representation}\label{sec:4_1}

Let
\[\mathcal{B} =\left\{ f_{\alpha} \right\}_{\alpha = 1}^{m}\]
denote a declared local structural representation containing \(m\) structural components. The index \(\alpha \in \left\{ 1,\ldots,m \right\}\) enumerates the fixed structural functions in \(\mathcal{B}\). Within this representation, write the local structural cost landscape as
\[V\left( x,\mathcal{B},\mathbf{J} \right) = \sum_{\alpha = 1}^{m}J_{\alpha}f_{\alpha}(x)\]
where
\[\mathbf{J} = \left\{ J_{\alpha} \right\}_{\alpha = 1}^{m}\]
contains the corresponding coefficients. The two ingredients are different kinds of objects: each \(f_{\alpha}(x)\) is a prescribed function of configuration, while \(J_{\alpha}\) is its numerical coefficient. The symbol \(\mathbf{J}\) denotes the complete vector of landscape-forming coefficients, not only pairwise network couplings. For example, for two Ising-type variables \(\sigma_{1},\sigma_{2} \in \{ - 1, + 1\}\), with \(\bm{\sigma} = (\sigma_{1},\sigma_{2})\), the interaction term \(f_{1}\left( \bm{\sigma} \right) = - \sigma_{1}\sigma_{2}\) may carry coefficient \(J_{1}\), while the local-field term \(f_{2}\left( \bm{\sigma} \right) = - \sigma_{1}\) may carry coefficient \(J_{2}\), which in the conventional notation of statistical physics would be written as a local field rather than a coupling; in that notation \(J_{1} = J_{12}\) and \(J_{2} = h_{1}\), and the minus signs are adopted so that a positive coupling favors spin alignment, while a positive local field favors the positive orientation of the corresponding spin. Both coefficients must be included when the whole landscape is rescaled. This example illustrates the distinction between a term and its coefficient, without restricting the general construction to spins, energy functions, or pairwise interactions.

At the current basin \(a\), the local structural state is represented by
\[S_{a} = \left( \mathcal{B},\mathbf{J} \right)\]
Both \(\mathcal{B}\) and \(\mathbf{J}\) are defined relative to the current basin. Their dependence on \(a\) is left implicit throughout.

When \(\mathcal{B}\) and \(\mathbf{J}\) are fixed, \(V(x,\mathcal{B},\mathbf{J})\) is abbreviated as \(V(x)\), as in Section~\ref{sec:3}. Differentiability is not required for the barrier-scaling result. When gradients and Hessians are used to describe the local geometry of a smooth landscape, \(V\) is assumed to be twice differentiable at the relevant configurations.

The representation \(\mathcal{B}\) and its normalization convention must be fixed for the comparison class before the magnitude of \(\mathbf{J}\) is used to define plasticity. This requirement removes a basic scaling ambiguity. For any \(c_{\alpha} > 0\),
\[f_{\alpha} \rightarrow c_{\alpha}f_{\alpha},\ J_{\alpha} \rightarrow \frac{J_{\alpha}}{c_{\alpha}}\]
leaves the product \(J_{\alpha}f_{\alpha}(x)\) unchanged. Without a fixed scale for each structural component, the same landscape could therefore be assigned different coefficient magnitudes.

Each component is normalized according to a declared convention,
\[\lVert f_{\alpha} \rVert_{\mathcal{N}} = 1\]
for every \(\alpha\), where the declared norm \(\lVert \cdot \rVert_{\mathcal{N}}\) is fixed for the model class.

Normalization removes arbitrary componentwise rescaling, but it does not make the decomposition unique. The same \(V\) may admit different normalized representations containing different structural components or different numbers of components. This non-uniqueness can arise for a single system. It is a general property of coefficient-based representations because the landscape alone does not select the family in which it must be decomposed.

Declaring \(\mathcal{B}\) therefore fixes the component family, its normalization, and, when coefficients are not uniquely identifiable, the estimation or identification rule used to obtain \(\mathbf{J}\). The ambiguity is resolved by a modeling convention for the comparison class, not by deriving a unique representation from \(V\).

Plasticity, quantified in Section~\ref{sec:4_2} as the inverse of the constraint magnitude, is consequently not a representation-free functional of \(V\) alone. It is defined relative to the declared structural representation:
\[P = P_{\mathcal{B}}\left( \mathbf{J} \right)\]
Different model classes may use different representations. Values obtained under different representations are representation-specific and are not directly comparable. No general mapping between admissible representations is supplied or claimed here.

Once \(\mathcal{B}\) has been fixed for a comparison class, however, \(P_{\mathcal{B}}\left( \mathbf{J} \right)\) is determined by the coefficient vector of each system rather than chosen separately for each comparison. In symptom-network applications, for example, the same item set, network specification, estimator, and normalization convention can place the estimated coefficients on a common representational basis (\protect\hyperlink{_ENREF_17}{Epskamp et al., 2018}).

This is the representation-relative sense in which \(P\) is a structural property of the system.

\subsection{Local constraint magnitude and plasticity}\label{sec:4_2}

For a fixed \(\mathcal{B}\), define the mean magnitude of local configurational constraint as
\[C_{\mathcal{B}}\left( \mathbf{J} \right) = \frac{1}{m}\sum_{\alpha = 1}^{m}\left| J_{\alpha} \right|\]
For \(C_{\mathcal{B}}\left( \mathbf{J} \right) > 0\), define plasticity as its inverse:
\[P_{\mathcal{B}}\left( \mathbf{J} \right) = \frac{1}{C_{\mathcal{B}}\left( \mathbf{J} \right)} = \frac{m}{\sum_{\alpha = 1}^{m} \lvert J_{\alpha}\rvert}\]
The zero-constraint case is the limiting value \(P \rightarrow \infty\) and is excluded from the finite comparisons below.

When the representation is unambiguous, its explicit dependence will be omitted and the quantities will be written as \(C\) and \(P\). Higher \(P\) means weaker mean configurational constraint within the fixed representation. Lower \(P\) means stronger mean constraint.

The magnitude functional records how much total coefficient magnitude the representation carries.

In the network realization, \(J_{k}\) and \(J_{kl}\) denote the coefficients attached to single elements and to pairs of elements, as a specific indexing of \(\mathbf{J}\). In network operationalizations, plasticity has been defined from aggregate edge weight (\protect\hyperlink{_ENREF_5}{Branchi, 2022}; \protect\hyperlink{_ENREF_6}{Branchi, 2023}; \protect\hyperlink{_ENREF_7}{Branchi, 2026})
\[P^{cross} = \frac{n}{\sum_{(k,l) \in \mathcal{E}_{a}}\left| J_{kl} \right|}\]
The denominator of this expression measures the magnitude of cross-coupling between elements. The magnitude functional defined here sums all landscape-forming coefficients. In a network, it therefore also includes the self-coupling \(J_{k}\) of each element, which constrains its displacement independently of the others
\[P^{aggregate} = \frac{n}{\sum_{k}\left| J_{k} \right| + \sum_{(k,l) \in \mathcal{E}_{a}}\left| J_{kl} \right|}\]
\(P^{aggregate}\) thus quantifies plasticity from both intrinsic and relational contributions to configurational constraint. Self-coupling is a static coefficient of the structural landscape, attached to a non-constant term of a single element, and is identified in the same way as cross-coupling. \(C\) is defined on the static landscape alone. Damping and other quantities of the dynamics do not enter \(C\).

The two measures coincide when self-coupling vanishes and nearly coincide when it is negligible. When the aggregate magnitude of self-coupling is constant across the systems compared, they preserve the same ordering, although their numerical values differ. Exact barrier scaling requires uniform rescaling of all landscape-forming coefficients. A change in cross-coupling alone, at fixed self-coupling, does not satisfy this condition. Both measures are normalized by the number of nodes \(n\), as in the companion formulation. With the component normalization of the general definition, the numerator becomes \(m_{all}\) (Section~\ref{sec:9_1}). The two conventions give the same ordering within a fixed representation but different values.

The factor \(m\) expresses constraint as a mean magnitude per represented structural component. It does not quantify degeneracy and performs no substantive work in the scaling theorem. Because \(m\) is fixed by \(\mathcal{B}\), it remains constant within every comparison to which the exact theorem applies.

The unsigned magnitude does not record whether signed coefficients can be satisfied jointly. In the symmetric signed-network realization of Section~\ref{sec:9_2}, effective aggregate coupling \(E_{a}^{\star}\) applies to the original edges. Including local fields as auxiliary edges yields \(E_{a,all}^{\star}\) for the joint contribution of self-coupling and cross-coupling.

More generally, the argument requires only a positive scalar magnitude functional satisfying
\[C_{\mathcal{B}}\left( \lambda\mathbf{J} \right) = \lambda C_{\mathcal{B}}\left( \mathbf{J} \right),\ \ \lambda > 0
\]
Positive homogeneity is sufficient for the scaling result, but does not by itself establish that a given functional is an appropriate measure of configurational constraint. Which functional is taken to measure constraint is a substantive choice within the declared representation, not a consequence of the scaling result. The normalized \(\ell^{1}\) coefficient magnitude is the simplest realization. Terms that are constant in the configuration are excluded from the declared representation, since adding a constant to the landscape changes no barrier and should not change the measured constraint. Excluding individually constant functions is not sufficient, because a combination of non-constant functions can itself be constant. The representation, its normalization, and the rule identifying the coefficients are therefore defined relative to a fixed convention that removes additive constants, including those generated by combinations of terms. If the normalized components are dimensionless, the coefficients, the constraint magnitude, the barriers, and the threshold share the units of the landscape, plasticity carries their inverse, and the normalized barriers are dimensionless; alternatively all these quantities may be expressed in common non-dimensionalized units.

\section{Constraint magnitude and normalized architecture}\label{sec:5}

Within a fixed representation, write
\[\mathbf{J} = C\widehat{\mathbf{J}}\]
where \(\widehat{\mathbf{J}}\) is normalized so that
\[\frac{1}{m}\sum_{\alpha = 1}^{m} \lvert {\widehat{J}}_{\alpha}\rvert = 1\]
The scalar \(C\) specifies the magnitude of local configurational constraint. The normalized coefficient vector \(\widehat{\mathbf{J}}\) specifies its architecture, meaning the distribution and signs of the coefficients within \(\mathcal{B}\).

Since
\[P = \frac{1}{C}\]
varying plasticity while holding \(\mathcal{B}\) and \(\widehat{\mathbf{J}}\) fixed changes the overall coefficient magnitude without changing the structural representation or normalized architecture. The family
\[\mathbf{J} = C\widehat{\mathbf{J}}\]
is a local structural ray because every member is a positive scalar multiple of the same normalized coefficient vector.

An affine transformation of estimated coupling coefficients has previously been used in energy-landscape analysis to locate empirical neural dynamics relative to a phase transition (\protect\hyperlink{_ENREF_19}{Ezaki et al., 2020}). When the additive offset vanishes, the coupling coefficients undergo radial rescaling. However, that analysis holds field terms fixed, so the full landscape is not generally rescaled uniformly. The exact structural ray considered here instead rescales every landscape-forming coefficient, allowing the local barrier spectrum to be ordered analytically.

Three objects play distinct roles:
\[\varphi\]
specifies which configurations count as functionally distinct alternatives,
\[\mathcal{B}\]
specifies how the local structural landscape is represented, and
\[\widehat{\mathbf{J}}\]
specifies the normalized architecture of the coefficients within that representation.

For the exact barrier ordering derived below, two systems or two states must share
\[\varphi, \mathcal{B}, \widehat{\mathbf{J}}\]
These conditions do not assume the result. They isolate variation in constraint magnitude from variation in functional resolution, representation, or architecture.

The fixed-architecture condition restricts the exact barrier ordering, not the definition of plasticity. An arbitrary change in \(\mathbf{J}\) can be decomposed into a radial change in \(C\) and an architectural change in \(\widehat{\mathbf{J}}\). The theorem determines the effect of the radial component. The architectural component may reinforce or oppose that effect and must be evaluated separately. Empirical comparisons approximate the exact radial case to the extent that their normalized architectures are similar. Architectural mismatch should therefore be quantified. Proposition~\ref{proposition:1} (Section~\ref{sec:7}) bounds its effect on the barrier spectrum.

\subsection{Architectural similarity in biological systems}\label{sec:5_1}

Within a common structural representation, plasticity can be compared across individuals even when their normalized architectures differ. However, the ordering of structural barriers also depends on the normalized architecture \(\widehat{\mathbf{J}}\), since architectural differences may reinforce or counteract the effects of plasticity.

In biological systems, architectural variation is not unconstrained. Shared genetic, developmental, physiological, and functional constraints render many architectures unrealizable and may confine the members of a species or population to a restricted region of that space. Individuals may therefore differ more in constraint magnitude \(C\) than in normalized architecture \(\widehat{\mathbf{J}}\). Functional brain networks illustrate this possibility. They exhibit shared organizational features and stable individual differences, with smaller contributions from task state and day-to-day variation (\protect\hyperlink{_ENREF_21}{Gratton et al., 2018}). Proposition~\ref{proposition:1} (Section~\ref{sec:7}) quantifies architectural differences and establishes a sufficient condition for preserving the ordering of positive finite structural barriers despite such differences.

\section{Positive homogeneity of the local landscape}\label{sec:6}

\begin{lemma}[Positive homogeneity of structural barriers]\label{lemma:1}

Fix the functional map \(\varphi\), the configurational space \(\mathcal{X}\), the admissible-path rules, and a structural representation \(\mathcal{B} =\left\{ f_{\alpha} \right\}_{\alpha = 1}^{m}\). Suppose that every landscape-forming term belongs to the coefficient vector being rescaled:
\[V\left( x,\mathcal{B},\mathbf{J} \right) = \sum_{\alpha = 1}^{m}{J_{\alpha}f_{\alpha}(x)}\]
Assume that the occupied basin \(a\) and its selected minimum \(x_{a}^{\star}\) are identified by rules invariant under positive rescaling of the landscape. Candidate alternatives are defined from the fixed functional map and the occupied basin, and admissible-path families are determined by the fixed configurational space and the declared path rules. Then, for every \(\lambda > 0\):

(i) the landscape is positively homogeneous,
\[V\left( x,\mathcal{B},\lambda\mathbf{J} \right) = \lambda V\left( x,\mathcal{B},\mathbf{J} \right)\]
(ii) the occupied basin, its selected minimum, the candidate set \(\mathcal{Y}_{a}\), and each admissible-path family \(\Gamma_{a\mu}\) remain unchanged;

(iii) every structural barrier satisfies
\[\Delta_{a\mu}\left( \mathcal{B},\lambda\mathbf{J} \right) = \lambda\Delta_{a\mu}\left( \mathcal{B},\mathbf{J} \right)\]
with the convention \(\lambda\infty = \infty\).
\end{lemma}

\begin{proof}
Part (i) follows immediately from linearity in the coefficients. Positive multiplication preserves the ordering of all structural costs. Together with the stated invariance of the basin-selection and path rules, this establishes (ii). For every \(\mu \in \mathcal{Y}_{a}\),
\[\Delta_{a\mu}\left( \mathcal{B},\lambda\mathbf{J} \right) = \inf_{\gamma \in \Gamma_{a\mu}}\left\lbrack \sup_{s \in I_{\gamma}}\lambda V\left( \gamma(s),\mathcal{B},\mathbf{J} \right) - \lambda V\left( x_{a}^{\star},\mathcal{B},\mathbf{J} \right) \right\rbrack\]
Because \(\lambda > 0\), multiplication by \(\lambda\) commutes with both the supremum along each path and the infimum over admissible paths, so that \(\Delta_{a\mu}\left( \mathcal{B},\lambda\mathbf{J} \right) = \lambda\Delta_{a\mu}\left( \mathcal{B},\mathbf{J} \right)\). The equality also holds when \(\Gamma_{a\mu} = \varnothing\).
\end{proof}

The barrier-scaling result requires neither differentiability nor the existence of separating saddles. It applies to continuous and discrete configurational spaces, provided that the reference minimum and the admissible-path families remain invariant under positive rescaling. The argument depends on the minimax construction of the barriers, not on a prescribed dynamical law.

\begin{remark}[Local structure under radial rescaling]\label{remark:1}

Positive multiplication preserves the ordering of structural costs and hence all local minima defined through that ordering. For a differentiable landscape, \(\nabla(\lambda V) = \lambda\nabla V\) and, when the Hessian exists, \(H_{\lambda\mathbf{J}} = \lambda H_{\mathbf{J}}\), where \(H_{\mathbf{J}}\) denotes the Hessian of \(V(x,\mathcal{B},\mathbf{J})\) with respect to \(x\). Thus, critical points retain their locations and Hessian eigenvalue signs. In particular, nondegenerate minima and separating saddles retain their classification. In a discrete space, local minima and basin constructions based solely on cost ordering and fixed tie-breaking rules are likewise preserved.
\end{remark}

The assumption that all landscape-forming terms are rescaled is essential. If the landscape contains an additive term whose coefficient is held fixed, the critical points can move and exact proportionality need not hold. Such a comparison does not lie on the structural ray defined above.

\subsection{Coefficient variation and representation variation}\label{sec:6_1}

A present-state contribution belongs to \(\mathbf{J}\) when it changes the magnitude of one or more coefficients multiplying fixed components in \(\mathcal{B}\). Its mechanistic origin does not matter. Intrinsic and relational contributions can both belong to \(\mathbf{J}\).

A feature that changes a shape function,
\[f_{\alpha}(x) \rightarrow f_{\alpha}^{'}(x)\]
changes \(\mathcal{B}\). For example, changing \(J\) in the product \(Jf(x)\) changes its coefficient. Replacing the fixed shape function \(f(x)\) with \(f(Jx)\) instead changes the function itself and generally defines a different structural representation or structural class.

The relevant distinction is therefore between variation within a fixed representation and variation of the representation itself. The exact theorem concerns the radial subset of the former, in which \(\widehat{\mathbf{J}}\) is also fixed.

\section{Plasticity and the local barrier spectrum}\label{sec:7}

For a system on a fixed structural ray,
\[\mathbf{J} = C\widehat{\mathbf{J}}\]
Lemma~\ref{lemma:1} gives
\[\Delta_{a\mu}\left( \mathcal{B},\mathbf{J} \right) = C\Delta_{a\mu}\left( \mathcal{B},\widehat{\mathbf{J}} \right)\]
Define the normalized reference barrier (a derived property of the landscape, not a structural term or coefficient) by
\[{\widehat{\Delta}}_{a\mu} = \Delta_{a\mu}\left( \mathcal{B},\widehat{\mathbf{J}} \right)\]
For fixed \(\varphi\), \(\mathcal{B}\), and \(\widehat{\mathbf{J}}\), the quantities \({\widehat{\Delta}}_{a\mu}\) are constants indexed by the same functional alternatives. Therefore,
\[\Delta_{a\mu} = C{\widehat{\Delta}}_{a\mu}\]
and since \(C = 1/P\),
\[\Delta_{a\mu} = \frac{{\widehat{\Delta}}_{a\mu}}{P}\]
This relation introduces no criterion of accessibility and assumes no prior relation between plasticity and repertoire size.

\begin{theorem}[Plasticity orders the complete local barrier spectrum componentwise]\label{theorem:1}
Consider two systems, or two states of one system, that share the same configurational space, occupy corresponding current basins, and satisfy
\[\varphi^{(1)} = \varphi^{(2)},\ \ \mathcal{B}^{(1)} = \mathcal{B}^{(2)},\ \ {\widehat{\mathbf{J}}}^{(1)} = {\widehat{\mathbf{J}}}^{(2)}\]
The admissible-path rules must also be the same in both systems, including the declared elementary moves in a discrete configurational space. If
\[P^{(2)} > P^{(1)}\]
then the barrier family is ordered componentwise: for every candidate alternative \(\mu\),
\[\Delta_{a\mu}^{(2)} = \frac{P^{(1)}}{P^{(2)}}\Delta_{a\mu}^{(1)} \leq \Delta_{a\mu}^{(1)}\]
\end{theorem}

\begin{proof}
The matching conditions ensure that both systems share the same candidate alternatives and normalized reference barriers. Since \(\Delta_{a\mu}^{(i)} = {\widehat{\Delta}}_{a\mu}/P^{(i)}\), the displayed equality follows. Because \(0 < P^{(1)}/P^{(2)} < 1\), the inequality is strict for every finite positive barrier. Zero and infinite barriers remain unchanged.
\end{proof}

The matching conditions make the two landscapes positive multiples of one another, preserving their basin structure. The correspondence of the occupied basins then ensures that the two barrier families are indexed by the same functional alternatives and share the same constants \({\widehat{\Delta}}_{a\mu}\).

Here componentwise means alternative by alternative. Each positive finite entry of the barrier family \(\left( \Delta_{a\mu} \right)_{\mu \in \mathcal{Y}_{a}}\) strictly decreases, each zero entry remains zero, and each infinite entry remains infinite. It does not mean that every coefficient in \(\mathbf{J}\) is compared independently.

The theorem is threshold-free. Increasing plasticity cannot raise the structural barrier to any candidate alternative within the matched local structural family. All positive finite barriers are rescaled by the same inverse factor.

Radial rescaling therefore preserves the ordering of the complete barrier spectrum and the ratios between its finite positive entries. It changes their common scale without changing their relative structure.

Plasticity is defined from the coefficients alone. Barrier height is not part of the definition. It is the quantity that the definition is shown to order. This direction prevents circularity.

The scaling argument rests on linearity of the landscape in its coefficients, positive homogeneity of the constraint magnitude, and invariance of the reference basin and admissible-path families under radial rescaling. Within this matched structural class, it applies independently of the substantive interpretation of the configurations. Theorem~\ref{theorem:1} expresses the same result in the coordinate proposed here as plasticity. The conceptual contribution is the identification of inverse configurational constraint as that coordinate.

\begin{proposition}[Barrier ordering under architectural differences]\label{proposition:1}
Consider two systems sharing the same configurational space, functional partition \(\varphi\), structural representation \(\mathcal{B}\), admissible-path rules, candidate alternatives \(\mathcal{Y}_{a}\), and reference minimum \(x_{a}^{\star}\), but possibly differing in normalized architecture \(\widehat{\mathbf{J}}\). Let \(r_{\alpha}\) bound the variation of each structural component \(f_{\alpha}\) between any two configurations, and define

\[\varepsilon = \sum_{\alpha = 1}^{m}\left| {\widehat{J}}_{\alpha}^{(2)} - {\widehat{J}}_{\alpha}^{(1)} \right|r_{\alpha}\]
For every alternative \(\mu\) with finite barriers,
\[\left| {\widehat{\Delta}}_{a\mu}^{(2)} - {\widehat{\Delta}}_{a\mu}^{(1)} \right| \leq \varepsilon\]
Consequently, for any alternative with \(0 < {\widehat{\Delta}}_{a\mu}^{(1)} < \infty\),
\[\frac{P^{(2)}}{P^{(1)}} > 1 + \frac{\varepsilon}{{\widehat{\Delta}}_{a\mu}^{(1)}}\quad \Rightarrow \quad\Delta_{a\mu}^{(2)} < \Delta_{a\mu}^{(1)}\]
\end{proposition}

\begin{proof}
Let \(h = {\widehat{V}}^{(2)} - {\widehat{V}}^{(1)}\), where \({\widehat{V}}^{(i)}(x) = \sum_{\alpha}{\widehat{J}}_{\alpha}^{(i)}f_{\alpha}(x)\). By the definition of \(\varepsilon\), for every configuration \(x\),

\[\left| h(x) - h\left( x_{a}^{\star} \right) \right| \leq \varepsilon\]
Thus, along every admissible path, the maximum structural cost relative to \(x_{a}^{\star}\) differs by at most \(\varepsilon\) between the two normalized landscapes. Taking the infimum over the same admissible paths preserves this bound. Since \(\Delta_{a\mu}^{(i)} = {\widehat{\Delta}}_{a\mu}^{(i)}/P^{(i)}\), the stated condition gives
\[\Delta_{a\mu}^{(2)} \leq \frac{{\widehat{\Delta}}_{a\mu}^{(1)} + \varepsilon}{P^{(2)}} < \frac{{\widehat{\Delta}}_{a\mu}^{(1)}}{P^{(1)}} = \Delta_{a\mu}^{(1)}\]
\end{proof}

When \(\varepsilon = 0\), the exact ordering of Theorem~\ref{theorem:1} is recovered. The required relative increase in plasticity grows with the architectural difference and decreases with the reference barrier. Ordering of all positive finite barriers is guaranteed by using the smallest positive finite normalized barrier, when it exists.

Proposition~\ref{proposition:1} provides a sufficient, not necessary, condition. It requires a shared reference minimum and bounded structural components. These conditions hold in finite discrete models with a common reference minimum, but require additional justification in continuous systems. In particular, the bound does not apply directly to Gaussian representations on unbounded configurational spaces. Theorem~\ref{theorem:1} is unaffected, since it does not require bounded structural components.

\section{From barrier spectrum to accessible repertoire}\label{sec:8}

A structural landscape specifies configurational costs. It does not by itself determine which costs should count as sufficiently low to be called accessible.

Introduce a fixed structural accessibility criterion

\[\theta > 0\]
The threshold is a barrier difference measured from the minimum \(x_{a}^{\star}\) of the current basin, not a separate threshold above each destination minimum. Define
\[\mathcal{R}_{a}(P,\theta) = \left\{ \mu \in \mathcal{Y}_{a}:\Delta_{a\mu} \leq \theta \right\}\]
Theorem~\ref{theorem:1} requires matched \(\varphi\), \(\mathcal{B}\), and \(\widehat{\mathbf{J}}\). Comparison of discrete repertoires additionally requires the same \(\theta\) in common structural-cost units. Matching \(\varphi\) gives the same functional alternatives, matching \(\mathcal{B}\) gives the coefficients the same structural meaning, and matching \(\widehat{\mathbf{J}}\) holds the normalized architecture fixed.

Radial variation of the coefficient magnitude must be distinguished from a mere change of cost units. A change of units would also alter the numerical value assigned to \(\theta\), leaving the accessible repertoire unchanged. Here the threshold is held fixed in common structural-cost units while the landscape is rescaled.

Using \(\Delta_{a\mu} = {\widehat{\Delta}}_{a\mu}/P\), the accessible repertoire can equivalently be written as
\[\mathcal{R}_{a}(P,\theta) = \left\{ \mu \in \mathcal{Y}_{a}:{\widehat{\Delta}}_{a\mu} \leq P\theta \right\}\]
Thus,
\[\mu \in \mathcal{R}_{a}(P,\theta)\  \Leftrightarrow \ P \geq \frac{{\widehat{\Delta}}_{a\mu}}{\theta}\]
Define
\[P_{a\mu}^{\star} = \frac{{\widehat{\Delta}}_{a\mu}}{\theta}\]
For a finite positive normalized barrier, this is the plasticity value at which alternative \(\mu\) first satisfies the declared structural-accessibility criterion. An alternative with zero barrier is accessible for every \(P > 0\), whereas one with infinite barrier remains inaccessible for every finite \(P\). The critical value is derived from barrier scaling rather than postulated as a definition of plasticity.

\begin{corollary}[Accessible repertoires are nested]\label{corollary:1}

Consider two systems satisfying the matching conditions of Theorem~\ref{theorem:1} and using a common threshold \(\theta\). If
\[P^{(2)} > P^{(1)}\]
then
\[\mathcal{R}_{a}\left( P^{(1)},\theta \right) \subseteq \mathcal{R}_{a}\left( P^{(2)},\theta \right)\]
Consequently,
\[\lvert \mathcal{R}_{a}\left( P^{(2)},\theta \right)\rvert \geq \lvert \mathcal{R}_{a}\left( P^{(1)},\theta \right)\rvert\]
\end{corollary}

\begin{proof}
If \(\mu \in \mathcal{R}_{a}\left( P^{(1)},\theta \right)\), then \(\Delta_{a\mu}^{(1)} \leq \theta\). By Theorem~\ref{theorem:1}, \(\Delta_{a\mu}^{(2)} \leq \Delta_{a\mu}^{(1)} \leq \theta\), so \(\mu \in \mathcal{R}_{a}\left( P^{(2)},\theta \right)\).
\end{proof}

For a finite repertoire, the inequality in repertoire size is strict whenever at least one positive barrier crosses \(\theta\) between the two plasticity values. For fixed \(\varphi\), \(\mathcal{B}\), \(\widehat{\mathbf{J}}\), and \(\theta\), the map

\[P \mapsto \lvert \mathcal{R}_{a}(P,\theta)\rvert\]
is therefore a non-decreasing step function whenever the set of functionally distinct alternatives is finite.

\begin{corollary}[Repertoire nesting under architectural differences]\label{corollary:2}

Consider two systems satisfying the conditions of Proposition~\ref{proposition:1} and using a common threshold \(\theta\), and let \(\mathcal{R}_{a}^{(i)}(\theta) = \left\{ \mu \in \mathcal{Y}_{a}:\Delta_{a\mu}^{(i)} \leq \theta \right\}\). If
\[P^{(2)} - P^{(1)} \geq \frac{\varepsilon}{\theta}\]
then
\[\mathcal{R}_{a}^{(1)}(\theta) \subseteq \mathcal{R}_{a}^{(2)}(\theta)\]
\end{corollary}

\begin{proof}
If \(\mu \in \mathcal{R}_{a}^{(1)}(\theta)\), then \({\widehat{\Delta}}_{a\mu}^{(1)} \leq P^{(1)}\theta\), which is finite. By Proposition~\ref{proposition:1}, \({\widehat{\Delta}}_{a\mu}^{(2)} \leq P^{(1)}\theta + \varepsilon \leq P^{(2)}\theta\), so \(\mu \in \mathcal{R}_{a}^{(2)}(\theta)\).
\end{proof}

When \(\varepsilon = 0\), Corollary~\ref{corollary:2} reduces to Corollary~\ref{corollary:1}.

\subsection{Worked scaling example}\label{sec:8_1}

Suppose a normalized architecture has three structural coefficients

\[\widehat{\mathbf{J}} = \left( 0.5,1.0,1.5 \right)\]
and a normalized barrier spectrum
\[\left( {\widehat{\Delta}}_{a\mu_{1}},{\widehat{\Delta}}_{a\mu_{2}},{\widehat{\Delta}}_{a\mu_{3}} \right) = \left( 1.0,0.4,1.8 \right)\]
A system with \(C = 2\) has \(P = 0.5\), coefficient vector \(\mathbf{J} = \left( 1,2,3 \right)\), and barriers \(\left( 2.0,0.8,3.6 \right)\). A second system on the same ray with \(C = 1\) has \(P = 1\), coefficient vector \(\mathbf{J} = \left( 0.5,1,1.5 \right)\), and barriers \(\left( 1.0,0.4,1.8 \right)\). Doubling plasticity halves every positive barrier without changing the alternatives or their ordering. If \(\theta = 1.2\), the accessible repertoire expands from \(\left\{ \mu_{2} \right\}\) to \(\left\{ \mu_{1},\mu_{2} \right\}\). The third alternative remains inaccessible. Figure~\ref{fig:1} schematically illustrates the same uniform barrier rescaling and the resulting expansion of the accessible repertoire.

This example illustrates both results. The theorem rescales barrier heights without imposing an accessibility threshold. The corollary records which entries cross an independently chosen threshold.

\begin{figure}[t]
\centering
\includegraphics[width=\textwidth]{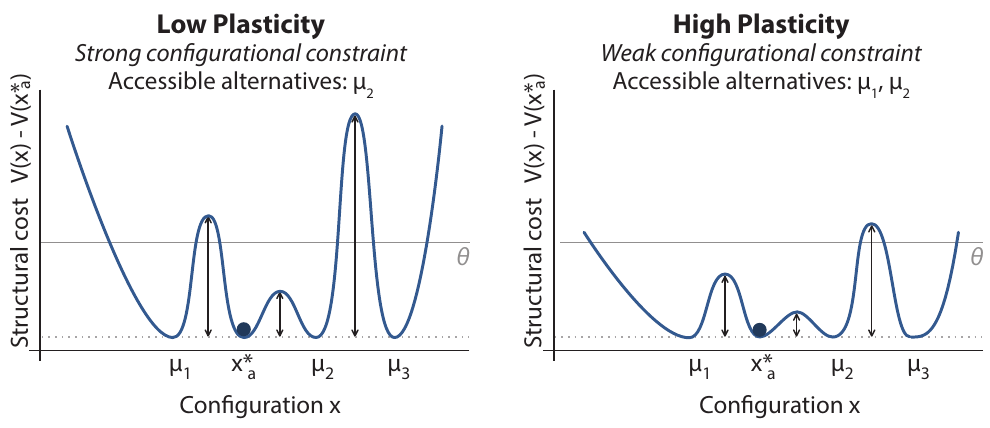}
\caption[Plasticity uniformly rescales the local barrier spectrum]{\textbf{Plasticity uniformly rescales the local barrier spectrum and can expand the accessible repertoire.} Schematic local structural cost landscapes $V(x)$ for two systems, or two states of one system, that satisfy the matching conditions of Theorem~\ref{theorem:1} and differ only in constraint magnitude $C = 1/P$. The filled circle marks the current-basin minimum $x_{a}^{\star}$. The alternatives $\mu_{1}$, $\mu_{2}$ and $\mu_{3}$ are functionally distinct. Vertical arrows indicate the structural barriers $\Delta_{a\mu}$ and the horizontal gray line indicates the fixed accessibility criterion $\theta$. The current-basin minimum is assigned zero cost in each panel. Both the barriers and $\theta$ are measured relative to that minimum. Destination minima are drawn at the same level for visual simplicity only. Left: low plasticity. Only $\mu_{2}$ satisfies the criterion. Right: plasticity doubled, $P^{(2)} = 2P^{(1)}$. Every positive finite barrier is halved and $\mu_{1}$ also becomes accessible (Theorem~\ref{theorem:1}, Corollary~\ref{corollary:1}). Barriers and $\theta$ are drawn to scale. See Section~\ref{sec:8} for details.}\label{fig:1}
\end{figure}

\section{Network realization}\label{sec:9}

\subsection{Coupling magnitude as configurational constraint}\label{sec:9_1}

A network provides an important realization of the general theory (\protect\hyperlink{_ENREF_18}{Ezaki et al., 2017}; \protect\hyperlink{_ENREF_28}{Masuda et al., 2025}; \protect\hyperlink{_ENREF_39}{Watanabe et al., 2014}). Suppose the components in \(\mathcal{B}\) represent specified interactions among network elements and the coefficients \(J_{\alpha}\) represent their coupling strengths. In this realization, the structural components are indexed directly by the edges \(\left( k,l \right) \in \mathcal{E}_{a}\), where \(\mathcal{E}_{a}\) denotes the edge set and \(J_{kl}\) is the coefficient attached to edge \(\left( k,l \right)\). Thus, \(m = \lvert \mathcal{E}_{a}\rvert\), and the component-normalized definition becomes
\[C^{edge} = \frac{1}{m}\sum_{(k,l) \in \mathcal{E}_{a}}\left| J_{kl} \right|,\ \ P^{edge} = \frac{m}{\sum_{(k,l) \in \mathcal{E}_{a}}\left| J_{kl} \right|}\]
The network interpretation does not make connectivity the \emph{definiendum}. Connectivity is one coefficient-level realization of configurational constraint when the declared representation assigns that role to network interactions.

An exact application of Theorem~\ref{theorem:1} to an edge-based network measure requires every term that shapes the relevant local barriers to belong to the coefficient structure being rescaled. If intrinsic or external-field terms remain fixed while only edge weights change, coupling-based plasticity can still be studied, but the exact proportionality between barriers and \(1/P\) is no longer guaranteed.

The node-normalized expression introduced in the companion network formulation is
\[P^{cross} = \frac{n}{\sum_{(k,l) \in \mathcal{E}_{a}}\left| J_{kl} \right|}\]
where \(n\) is the number of network nodes (\protect\hyperlink{_ENREF_7}{Branchi, 2026}). The two normalizations are related by
\[P^{cross} = \frac{n}{m}P^{edge}\]
They therefore produce the same ordering and the same radial scaling whenever the representation and network size are fixed, because \(n/m\) is then constant. They are not numerically identical across representations or network sizes. The numerator \(m\) in the present construction averages over represented structural contributions. The numerator \(n\) in the companion formulation encodes node-based configurational dimensionality. Cross-size comparisons depend on that model-specific choice and are not supplied by the exact theorem developed here.

Self-coupling enters the network realization as defined in Section~\ref{sec:4_2}. In Gaussian representations, self-coupling corresponds to the diagonal of the precision matrix. In Ising representations with symmetric coding, the diagonal term is constant and self-coupling is carried by the local field. In the component-normalized form, the count becomes \(m_{all} = n + m_{cross}\), with \(m_{cross} = \left| \mathcal{E}_{a} \right|\), when each node carries one self-coupling term.

\subsection{Effective and nominal aggregate coupling}\label{sec:9_2}

Effective aggregate coupling is the aggregate coupling that constrains the represented system. Nominal aggregate coupling is the aggregate unsigned magnitude of the direct couplings among its elements. They can differ because unrepresented common causes can add induced coupling to the represented coefficients, whereas frustration can reduce how much of the unsigned magnitude of those coefficients is jointly realizable as configurational constraint. The general definition \(P = 1/C\) is unchanged. The two cases concern what the represented coefficients contain and which aggregate functional of them instantiates the configurational constraint \(C\).

\emph{Increase by common causes.} The coefficients entering \(C\) are those of the declared structural representation, and the exact results apply only to coefficients that enter \(V\) in that representation. Their origin need not be a direct pairwise causal interaction. In statistical network representations, dependence induced by an unrepresented common cause can appear as coupling among the represented elements. For continuous, normally distributed variables, a unidimensional latent-variable model corresponds to a fully connected Gaussian graphical model, so that common variance remains represented in the network coefficients (\protect\hyperlink{_ENREF_37}{Waldorp and Marsman, 2022}). For binary variables, Ising network models are likewise statistically equivalent to a class of latent-variable models from item response theory (\protect\hyperlink{_ENREF_27}{Marsman et al., 2018}). A common cause can therefore raise the unsigned magnitude of the represented coefficients above nominal aggregate coupling, even in the absence of direct interactions among the affected elements. For inference about direct causal interactions, such induced coupling reflects confounding by the common cause. For plasticity, within the declared representation, such induced coupling remains constraint: the induced contributions enter \(\mathbf{J}\) and therefore contribute to \(V\) and shape its barriers in the same way as other contributions. A given pattern of dependence can be compatible with direct interactions, a common cause, or both, so the dependence structure alone does not identify the causal mechanism (\protect\hyperlink{_ENREF_27}{Marsman et al., 2018}; \protect\hyperlink{_ENREF_34}{van der Maas et al., 2006}; \protect\hyperlink{_ENREF_37}{Waldorp and Marsman, 2022}). In an estimated network, the coefficients already include any induced part, so nominal aggregate coupling is not observed separately. Constraint is therefore read from the coupling of the declared representation: plasticity is defined relative to the represented system, not to the causal mechanism by which its coupling arose.

\emph{Reduction by frustration.} The unsigned magnitude of the pairwise coefficients measures how much coupling is represented, but not necessarily how much can be expressed jointly as configurational constraint. In the absence of induced coefficients, this unsigned magnitude is the nominal aggregate coupling. It ignores whether signed couplings can be satisfied jointly. In a signed network, frustration occurs when no orientation \(\sigma_{k} \in \left\{ -  1, + 1 \right\}\) satisfies \(J_{kl}\sigma_{k}\sigma_{l} = \lvert J_{kl}\rvert\) for every edge \(\left( k,l \right) \in \mathcal{E}_{a}\) (\protect\hyperlink{_ENREF_20}{Fischer and Hertz, 1991}; \protect\hyperlink{_ENREF_33}{Toulouse, 1977}). In the symmetric signed-network realization considered here, effective aggregate coupling of the cross-coupling coefficients is given by the maximum jointly realizable coupling
\[E_{a}^{\star} = \max_{\bm{\sigma} \in \left\{ - 1, + 1 \right\}^{n}}{\sum_{(k,l) \in \mathcal{E}_{a}}J_{kl}}\sigma_{k}\sigma_{l}\]
It is an aggregate functional of the coefficients, not a second set of edge weights. Finite comparisons using this functional require \(E_{a}^{\star} > 0\). For a structurally balanced signed network (\protect\hyperlink{_ENREF_8}{Cartwright and Harary, 1956}),
\[E_{a}^{\star} = \sum_{(k,l) \in \mathcal{E}_{a}}\left| J_{kl} \right|\]
Effective aggregate coupling and the unsigned magnitude then coincide. In general, their ratio is the normalized net coupling
\[s_{a} = \frac{E_{a}^{\star}}{\sum_{(k,l) \in \mathcal{E}_{a}}\left| J_{kl} \right|}\]
Replacing the unsigned magnitude by effective aggregate coupling \(E_{a}^{\star}\) in the denominator gives the frustration-adjusted quantities
\[P^{fr,edge} = \frac{m}{E_{a}^{\star}},\ \ P^{fr,cross} = \frac{n}{E_{a}^{\star}}\]
Since
\[E_{a}^{\star} = s_{a}\sum_{\left( k,l \right) \in \mathcal{E}_{a}} \lvert J_{kl}\rvert 
\]
the adjusted quantities can equivalently be written as
\[P^{fr,edge} = \frac{P^{edge}}{s_{a}},\ P^{fr,cross} = \frac{P^{cross}}{s_{a}}
\]
Thus, where couplings conflict, effective aggregate coupling is smaller than the unsigned magnitude of the coefficients and the corresponding plasticity is larger by the factor \(1/s_{a}\). The inverse relation that defines plasticity is preserved. The frustration adjustment replaces the unsigned coefficient magnitude with a jointly realizable aggregate functional within a fixed representation. A change in the shape functions or in the representation itself is not an adjustment of this kind and defines a different comparison class (Section~\ref{sec:6_1}).

Because
\[E_{a}^{\star}\left( \lambda\mathbf{J} \right) = \lambda E_{a}^{\star}\left( \mathbf{J} \right),\ \ \lambda > 0\]
the positive-homogeneity argument remains exact along a structural ray normalized by this functional. Finding the maximizing orientation becomes computationally difficult for large general signed networks (\protect\hyperlink{_ENREF_1}{Barahona, 1982}). That difficulty concerns estimation rather than definition. The same construction can include local fields in the symmetric signed-network realization. Introduce an auxiliary element \(0\) fixed at \(\sigma_{0} = + 1\), and augment the edge set with \((0,k)\), assigning \(J_{0k} = J_{k}\). Maximizing over the original nodes gives the jointly realizable contribution of all local fields and pairwise couplings
\[E_{a,all}^{\star} = \max_{\bm{\sigma} \in \left\{ - 1, + 1 \right\}^{n}}\left\lbrack \sum_{k = 1}^{n}{J_{k}\sigma_{k}} + \sum_{(k,l) \in \mathcal{E}_{a}}{J_{kl}\sigma_{k}\sigma_{l}} \right\rbrack\]
Finite comparisons require \(E_{a,all}^{\star} > 0\). The auxiliary element is fixed and adds neither a system node nor an admissible move. This augmentation rewrites local-field terms already included in the declared representation. It does not change the structural landscape or the comparison class. The corresponding node-normalized measure is \(P^{fr,aggregate} = n/E_{a,all}^{\star}\). Under component normalization, its numerator is \(m_{all} = n + \left| \mathcal{E}_{a} \right|\) when each node carries one local-field term. Because \(E_{a,all}^{\star}\left( \lambda\mathbf{J} \right) = \lambda E_{a,all}^{\star}\left( \mathbf{J} \right)\) for \(\lambda > 0\), the barrier-scaling result also applies to this measure when all field and pairwise coefficients are rescaled together.

\emph{Domain and coding.} Network elements can be coded symmetrically, with values in \(\left\{ -  1, + 1 \right\}\), or non-negatively, with values in \(\left\{ 0,1 \right\}\), where zero denotes absence. A symmetric coding need not imply a mirror state because its negative level can also denote absence. A coherent transformation of couplings, fields, and the constant term makes the two codings descriptions of the same model, but it is not a rescaling of the coefficient vector alone because it transfers weight between couplings and field terms. Theorem~\ref{theorem:1} therefore applies exactly only when every relevant field term is included in the rescaled vector. The codings are thus instances of the same theory rather than competing formulations.

\subsection{Relation to psychological network models}\label{sec:9_3}

When an estimated psychological network is summarized by aggregate absolute edge weight, network-based plasticity is its reciprocal up to the declared normalization. The contribution is therefore theoretical rather than arithmetical. Aggregate coupling is interpreted as present configurational constraint, its reciprocal defines plasticity, and the resulting consequences for structural barriers are derived. Edge weights capture cross-coupling only. Accordingly, their reciprocal approximates plasticity based on the full configurational constraint only when self-coupling is negligible. Otherwise, self-coupling must be included explicitly in the aggregate measure (Section~\ref{sec:4_2}).

This interpretation does not imply that aggregate connectivity alone determines barrier landscapes across models. Their geometry also depends on model type, variable encoding, thresholds, and other parameters (\protect\hyperlink{_ENREF_3}{Borsboom, 2017}; \protect\hyperlink{_ENREF_10}{Cramer et al., 2016}; \protect\hyperlink{_ENREF_25}{Hoekstra et al., 2026}). Theorem~\ref{theorem:1} states only that uniform rescaling of all coefficients in a fixed normalized architecture rescales the local structural barriers. Changes in encoding, thresholds, or coefficient architecture define different comparison classes.

In statistical Ising representations, the inverse temperature \(\beta\) multiplies the Hamiltonian and can be absorbed into the couplings and thresholds, so that varying \(\beta\) while holding the underlying couplings and thresholds fixed is equivalent to a uniform rescaling of these coefficients (\protect\hyperlink{_ENREF_11}{Dalege et al., 2018}; \protect\hyperlink{_ENREF_22}{Grimes et al., 2025}). Within that representation, temperature variation is therefore a special case of displacement along the structural ray of Theorem~\ref{theorem:1}, which extends such comparisons to any declared representation and specifies the conditions under which they order barriers and accessible repertoires. This equivalence is representational and does not identify plasticity with thermodynamic temperature.

\subsection{Instantiating the declared representation in estimated networks}\label{sec:9_4}

The definition becomes empirical when the declared representation is instantiated with data. In studies of mental states, the structural components in \(\mathcal{B}\) can represent specified interactions among a fixed set of variables, such as symptom items, affective states, or regional neural signals. Each function \(f_{\alpha}\) specifies the form of one such interaction, while the coefficients \(J_{\alpha}\) are the estimated interaction weights. When the estimation model identifies static on-site coefficients, these enter \(C\) as self-coupling. The unit diagonal of a correlation matrix, fixed at one by construction, and autoregressive coefficients, which belong to the dynamics, are not such coefficients. Under the normalized unsigned \(\ell^{1}\) realization introduced in Section~\ref{sec:4_2}, \(C\) is their mean absolute value and \(P = 1/C\). Because these weights are estimated from observed dependence, they can include coupling induced by unrepresented common causes (Section~\ref{sec:9_2}). Noise that is independent across elements generates no systematic covariance among them. It can affect the precision and magnitude of estimated coefficients, but the covariance it produces by chance diminishes as observations accumulate. In time-series data, a structured common driver can induce coupling when it varies over the observation window, which must therefore be declared.

The estimate is prospective when the interaction weights are inferred from observations available before the change of interest. These observations may consist of intensive repeated measurements of a single system, such as experience-sampling or resting-state time series, or baseline measurements from a defined population (\protect\hyperlink{_ENREF_18}{Ezaki et al., 2017}; \protect\hyperlink{_ENREF_29}{Myin-Germeys et al., 2018}). The unit of representation must therefore be stated explicitly. A person-specific network represents an individual, whereas a network estimated from pooled observations represents the defined population or subgroup. In the first case, \(a\) indexes the current basin of that individual. In the second, the landscape and its basins belong to the population-level representation and carry over to individuals only under additional assumptions. This network operationalization was proposed in previous theoretical work (\protect\hyperlink{_ENREF_5}{Branchi, 2022}; \protect\hyperlink{_ENREF_6}{Branchi, 2023}).

Direct comparisons of \(P\) require the representation to remain fixed across the systems being compared, which in this setting means a common item set or parcellation, a common estimator, and a common normalization convention. To connect the resulting estimates to barrier spectra or accessible repertoires, the functional map \(\varphi\) must also be declared. Whether two symptom profiles with the same total score count as one functional alternative or as two determines the repertoire and the indexing of the barrier family. These choices are not supplied by the theory and affect the substantive interpretation of the estimate. Prospective estimates of configurational constraint and plasticity within a declared representation do not, by themselves, test Theorem~\ref{theorem:1}, which requires a common normalized architecture. When architectures differ, Proposition~\ref{proposition:1} bounds their effect on structural barriers under its stated conditions.

\subsection{Anticipating changes in mental health from network connectivity}\label{sec:9_5}

A network-based operationalization of plasticity has been investigated in studies of depression and bipolar disorder. In line with the present hypothesis, in patients with major depressive disorder, group-level symptom-network connectivity was weaker in responders than in non-responders, and it was inversely related to improvement (\protect\hyperlink{_ENREF_14}{Delli Colli et al., 2024}). In patients experiencing a major depressive episode, symptom-network connectivity was weaker in bipolar disorder than in major depressive disorder, despite comparable depressive symptom severity (\protect\hyperlink{_ENREF_35}{Viglione et al., 2025}). In two clinical trials, lower baseline connectivity, and thus higher estimated plasticity, was found in patients who recovered faster. Plasticity anticipated time to recovery only in patients reporting a favorable context (\protect\hyperlink{_ENREF_15}{Delli Colli et al., 2025b}). Across groups of patients stratified by the magnitude of subsequent clinical change, baseline group-network connectivity was inversely related to that magnitude, whereas connectivity estimated during change tracked the change itself (\protect\hyperlink{_ENREF_12}{Delli Colli et al., 2025a}). This dissociation is the one expected of a prospective property. In healthy individuals, person-specific baseline connectivity was inversely associated with the magnitude of subsequent depressive symptom changes and predicted vulnerability to depressive symptomatology (\protect\hyperlink{_ENREF_13}{Delli Colli et al., 2026}).

\section{Biological implications and conceptual distinctions}\label{sec:10}

From a biological perspective, the framework separates three questions: how strongly a system is currently constrained, which functionally distinct alternatives are structurally accessible, and which changes actually occur. These questions need not have the same answer. Two systems with the same plasticity may have different barrier spectra if their normalized architectures differ, while similar observed changes do not imply similar plasticity. In applications to cellular, neural, or behavioral states, the functional alternatives and structural representation must therefore be specified independently of the changes subsequently observed. The theory establishes conditions under which structural barriers and accessible repertoires can be ordered. It does not determine the direction, timing, or probability of biological transitions.

\subsection{Plasticity is not realized change}\label{sec:10_1}

Plasticity is a prospective property of the system, defined before an outcome is realized. A highly plastic system may remain in its current state, whereas a weakly plastic system may undergo a large change along one of only a few available alternatives.

\subsection{Plasticity is not a rate}\label{sec:10_2}

Plasticity is not a temporal derivative such as
\[\frac{d\mathbf{J}}{dt}\]
A rate requires a trajectory. The present definition requires the current structural organization.

\subsection{Plasticity is not degeneracy}\label{sec:10_3}

Degeneracy is the multiplicity of distinct physical configurations that realize the same functional outcome. Plasticity concerns the constraint on access to functionally distinct outcomes. The two concepts are not identical, but degeneracy can affect configurational accessibility. Greater degeneracy may provide more paths to an alternative and thereby lower its minimax barrier. In the present construction, this effect enters \(\Delta_{a\mu}\) rather than the number of alternatives under \(\varphi\).

\subsection{Plasticity is not a transition probability}\label{sec:10_4}

The structural theorem assigns no probability to subsequent trajectories. Plasticity asks how strongly the present structure constrains access to available alternatives. Within the matched structural class, greater plasticity lowers structural barriers. When the dynamics is sensitive to these structural costs, this barrier lowering can increase transition propensities, but plasticity does not by itself determine which alternative will be realized or with what probability. Determining these probabilities additionally requires dynamics, noise, initial conditions, perturbations, and an observation horizon.

\subsection{Plasticity is not an entropy}\label{sec:10_5}

The accessible repertoire \(\lvert \mathcal{R}_{a}\rvert\) counts functionally distinct alternatives satisfying a declared structural criterion. It is not an entropy of a normalized outcome distribution. Entropy characterizes how probability or weight is distributed across possibilities. No entropy concept is required for the definition, theorem, or corollary.

\section{Scope and limitations}\label{sec:11}

\subsection{Locality}\label{sec:11_1}

The theory is local to the current configurational basin. Both \(P\) and the indexed barrier family
\[\left\{ \Delta_{a\mu} \right\}_{\mu \in \mathcal{Y}_{a}}\]
refer to departures from that basin. The theorem does not describe transitions among successive basins or multistep exploration of a global landscape.

\subsection{Structural costs rather than dynamical propensities}\label{sec:11_2}

No transition rate, noise amplitude, observation horizon, friction coefficient, inertia, or first-passage process enters the theorem. Accordingly, \(\Delta_{a\mu}\) orders structural costs. Once a dynamical law relating those costs to transitions is specified, they contribute to transition propensities but do not determine them by themselves. Two alternatives with the same barrier can have different transition rates because kinetic prefactors, path geometry, and stochastic forcing may differ (\protect\hyperlink{_ENREF_23}{Hänggi et al., 1990}; \protect\hyperlink{_ENREF_26}{Kramers, 1940}; \protect\hyperlink{_ENREF_28}{Masuda et al., 2025}).

\subsection{Representation-relative applicability}\label{sec:11_3}

Normalization removes arbitrary componentwise rescaling but does not select a unique decomposition. Plasticity is therefore defined relative to the declared structural representation \(\mathcal{B}\):
\[P = P_{\mathcal{B}}\left( \mathbf{J} \right)\]
Direct comparisons therefore require a common \(\mathcal{B}\) and coefficient-identification procedure, fixed independently of the magnitudes used to calculate \(P\). Under these conditions, \(P\) is directly comparable across systems. Section~\ref{sec:9_4} describes their implementation in estimated networks. Cross-representation comparisons require an explicit mapping.

\subsection{Matched functional resolution}\label{sec:11_4}

The functional map \(\varphi\) fixes which distinctions count as different alternatives. A finer partition can distinguish outcomes that a coarser partition groups together. Accordingly, direct comparisons of barrier spectra and accessible repertoires require matched \(\varphi\), unless an explicit mapping between partitions is supplied. Once \(\varphi\) is fixed, barrier indexing and repertoire cardinality are defined consistently across the comparison class.

\subsection{Fixed normalized architecture}\label{sec:11_5}

Theorem~\ref{theorem:1} applies to systems on the same local structural ray. When normalized architectures differ, plasticity alone does not determine the ordering of structural barriers. Proposition~\ref{proposition:1} provides a sufficient condition for preserving the ordering of positive finite barriers when the relative increase in plasticity is sufficiently large compared with architectural differences, provided that the systems share a reference minimum and have bounded structural components.

\subsection{Measurement}\label{sec:11_6}

The theory specifies a \emph{definiendum} and derives a structural consequence. It does not establish that a particular connectivity matrix, imaging signal, molecular marker, or behavioral assay directly measures \(\mathbf{J}\), \(C\), or \(P\).

Empirical estimation requires a measurement model linking observables to a declared representation and its coefficients. Reliability, estimation error, architectural mismatch, and the validity of the functional partition must then be evaluated. These requirements should be assessed against the theoretical property defined here rather than used to define plasticity retrospectively.

\section{Conclusion}\label{sec:12}

Plasticity is often treated as self-evident, yet most operational measures capture realized change rather than the capacity for change. I propose that the prospective property these measures seek to estimate is inverse configurational constraint. Relative to a declared structural representation,

\[P = \frac{1}{C}
\]
Higher plasticity therefore means that the present configuration imposes weaker constraint on displacement toward functionally distinct alternatives. Plasticity is specified by current structure before any subsequent trajectory is observed.

Within the exact structural class defined here, positive homogeneity implies that higher plasticity uniformly lowers every positive finite local barrier. Given a common structural-accessibility criterion, the accessible repertoire can therefore expand but cannot contract. These are derived consequences within the specified structural class, not parts of the definition. They hold when the functional partition, structural representation, normalized architecture, and admissible-path rules are fixed.

In network representations, aggregate coupling, comprising self-coupling and cross-coupling, provides a coefficient-level realization of configurational constraint, linking connectivity to a prospective capacity for change. Plasticity is not realized change, a rate of change, or a transition probability, although it can shape transition propensities through its effects on structural barriers. It captures how weakly the present structure constrains what the system could do otherwise.

For biology, the main consequence is that the capacity for change can be characterized from present structure before change occurs. In mental health, this provides a theoretical basis for investigating whether baseline network organization can help identify individuals with different capacities for subsequent change. The definition is nonetheless not tied to any particular system or level of analysis. It applies wherever the relevant constraints admit a declared structural representation, from molecular and cellular systems to neural, behavioral, and relational levels. It thereby provides a common theoretical language for understanding plasticity across biological scales and beyond biology, without conflating their distinct mechanisms.

\textbf{\hfill\break
}

\section*{Author contributions}

I.B.: conceptualization, formal analysis, methodology, visualization, writing---original draft, writing---review and editing.

\section*{Conflict of interest declaration}

I declare I have no competing interests.

\section*{Funding}

The research was supported by the grant from the ERANET Neuron and Istituto Superiore di Sanità, project EnviroMood, and the and Istituto Superiore di Sanità, project NeuroDynamics ISS20-6bbea6c1e02b, both to IB. The funders had no role in the preparation, review, or approval of the manuscript, or decision to submit the manuscript for publication.

\section*{Acknowledgments}

I thank Claudia Delli Colli and Jacopo Niedda for valuable discussion of the ideas presented in this paper.

\section*{References}
\begin{list}{}{\leftmargin=1.5em\itemindent=-1.5em\itemsep=2pt\parsep=0pt\small}
\item \hypertarget{_ENREF_1}{}Barahona, F., 1982. On the computational complexity of Ising spin glass models. Journal of Physics A: Mathematical and General 15, 3241--3253, \href{https://doi.org/10.1088/0305-4470/15/10/028}{doi:10.1088/0305-4470/15/10/028}.
\item \hypertarget{_ENREF_2}{}Bliss, T. V., Lomo, T., 1973. Long-lasting potentiation of synaptic transmission in the dentate area of the anaesthetized rabbit following stimulation of the perforant path. J Physiol 232, 331--56, \href{https://doi.org/10.1113/jphysiol.1973.sp010273}{doi:10.1113/jphysiol.1973.sp010273}.
\item \hypertarget{_ENREF_3}{}Borsboom, D., 2017. A network theory of mental disorders. World Psychiatry 16, 5--13, \href{https://doi.org/10.1002/wps.20375}{doi:10.1002/wps.20375}.
\item \hypertarget{_ENREF_4}{}Branchi, I., 2011. The double edged sword of neural plasticity: increasing serotonin levels leads to both greater vulnerability to depression and improved capacity to recover. Psychoneuroendocrinology 36, 339--51, \href{https://doi.org/10.1016/j.psyneuen.2010.08.011}{doi:10.1016/j.psyneuen.2010.08.011}.
\item \hypertarget{_ENREF_5}{}Branchi, I., 2022. Plasticity in mental health: A network theory. Neurosci Biobehav Rev 138, 104691, \href{https://doi.org/10.1016/j.neubiorev.2022.104691}{doi:10.1016/j.neubiorev.2022.104691}.
\item \hypertarget{_ENREF_6}{}Branchi, I., 2023. A mathematical formula of plasticity: Measuring susceptibility to change in mental health and data science. Neurosci Biobehav Rev 152, 105272, \href{https://doi.org/10.1016/j.neubiorev.2023.105272}{doi:10.1016/j.neubiorev.2023.105272}.
\item \hypertarget{_ENREF_7}{}Branchi, I., 2026. Quantifying plasticity: A network-based framework linking structure to dynamical regimes. Neurosci Biobehav Rev 187, 106765, \href{https://doi.org/10.1016/j.neubiorev.2026.106765}{doi:10.1016/j.neubiorev.2026.106765}.
\item \hypertarget{_ENREF_8}{}Cartwright, D., Harary, F., 1956. Structural balance: a generalization of Heider\textquotesingle s theory. Psychol Rev 63, 277--93, \href{https://doi.org/10.1037/h0046049}{doi:10.1037/h0046049}.
\item \hypertarget{_ENREF_9}{}Citri, A., Malenka, R. C., 2008. Synaptic plasticity: multiple forms, functions, and mechanisms. Neuropsychopharmacology 33, 18--41, \href{https://doi.org/10.1038/sj.npp.1301559}{doi:10.1038/sj.npp.1301559}.
\item \hypertarget{_ENREF_10}{}Cramer, A. O., van Borkulo, C. D., Giltay, E. J., van der Maas, H. L. J., Kendler, K. S., Scheffer, M., Borsboom, D., 2016. Major Depression as a Complex Dynamic System. PLoS One 11, e0167490, \href{https://doi.org/10.1371/journal.pone.0167490}{doi:10.1371/journal.pone.0167490}.
\item \hypertarget{_ENREF_11}{}Dalege, J., Borsboom, D., van Harreveld, F., van der Maas, H. L. J., 2018. The Attitudinal Entropy (AE) Framework as a General Theory of Individual Attitudes. Psychological Inquiry 29, 175--193, \href{https://doi.org/10.1080/1047840X.2018.1537246}{doi:10.1080/1047840X.2018.1537246}.
\item \hypertarget{_ENREF_12}{}Delli Colli, C., Viglione, A., Giuliani, A., Branchi, I., 2025a. Anticipating depression trajectories by measuring plasticity and change through symptom network dynamics. Eur Psychiatry 68, e128, \href{https://doi.org/10.1192/j.eurpsy.2025.10083}{doi:10.1192/j.eurpsy.2025.10083}.
\item \hypertarget{_ENREF_13}{}Delli Colli, C., Viglione, A., Poggini, S., Branchi, I., 2026. Network-based assessment of plasticity predicts vulnerability to depressive symptomatology in healthy individuals. J Affect Disord 412, 122129, \href{https://doi.org/10.1016/j.jad.2026.122129}{doi:10.1016/j.jad.2026.122129}.
\item \hypertarget{_ENREF_14}{}Delli Colli, C., Chiarotti, F., Campolongo, P., Giuliani, A., Branchi, I., 2024. Towards a network-based operationalization of plasticity for predicting the transition from depression to mental health. Nature Mental Health 2, 200--208, \href{https://doi.org/10.1038/s44220-023-00192-z}{doi:10.1038/s44220-023-00192-z}.
\item \hypertarget{_ENREF_15}{}Delli Colli, C., Viglione, A., Poggini, S., Cirulli, F., Chiarotti, F., Giuliani, A., Branchi, I., 2025b. A network-based analysis anticipates time to recovery from major depression revealing a plasticity by context interplay. Transl Psychiatry 15, 32, \href{https://doi.org/10.1038/s41398-025-03246-1}{doi:10.1038/s41398-025-03246-1}.
\item \hypertarget{_ENREF_16}{}Edelman, G. M., Gally, J. A., 2001. Degeneracy and complexity in biological systems. Proc Natl Acad Sci U S A 98, 13763--8, \href{https://doi.org/10.1073/pnas.231499798}{doi:10.1073/pnas.231499798}.
\item \hypertarget{_ENREF_17}{}Epskamp, S., Borsboom, D., Fried, E. I., 2018. Estimating psychological networks and their accuracy: A tutorial paper. Behav Res Methods 50, 195--212, \href{https://doi.org/10.3758/s13428-017-0862-1}{doi:10.3758/s13428-017-0862-1}.
\item \hypertarget{_ENREF_18}{}Ezaki, T., Watanabe, T., Ohzeki, M., Masuda, N., 2017. Energy landscape analysis of neuroimaging data. Philos Trans A Math Phys Eng Sci 375, 20160287, \href{https://doi.org/10.1098/rsta.2016.0287}{doi:10.1098/rsta.2016.0287}.
\item \hypertarget{_ENREF_19}{}Ezaki, T., Fonseca Dos Reis, E., Watanabe, T., Sakaki, M., Masuda, N., 2020. Closer to critical resting-state neural dynamics in individuals with higher fluid intelligence. Commun Biol 3, 52, \href{https://doi.org/10.1038/s42003-020-0774-y}{doi:10.1038/s42003-020-0774-y}.
\item \hypertarget{_ENREF_20}{}Fischer, K. H., Hertz, J. A., 1991. Spin Glasses. Cambridge University Press, Cambridge.
\item \hypertarget{_ENREF_21}{}Gratton, C., Laumann, T. O., Nielsen, A. N., Greene, D. J., Gordon, E. M., Gilmore, A. W., Nelson, S. M., Coalson, R. S., Snyder, A. Z., Schlaggar, B. L., Dosenbach, N. U. F., Petersen, S. E., 2018. Functional Brain Networks Are Dominated by Stable Group and Individual Factors, Not Cognitive or Daily Variation. Neuron 98, 439--452.e5, \href{https://doi.org/10.1016/j.neuron.2018.03.035}{doi:10.1016/j.neuron.2018.03.035}.
\item \hypertarget{_ENREF_22}{}Grimes, P. Z., Murray, A. L., Smith, K., Allegrini, A. G., Piazza, G. G., Larsson, H., Epskamp, S., Whalley, H. C., Kwong, A. S. F., 2025. Network temperature as a metric of stability in depression symptoms across adolescence. Nature Mental Health 3, 548--557, \href{https://doi.org/10.1038/s44220-025-00415-5}{doi:10.1038/s44220-025-00415-5}.
\item \hypertarget{_ENREF_23}{}Hänggi, P., Talkner, P., Borkovec, M., 1990. Reaction-rate theory: fifty years after Kramers. Reviews of Modern Physics 62, 251--341, \href{https://doi.org/10.1103/RevModPhys.62.251}{doi:10.1103/RevModPhys.62.251}.
\item \hypertarget{_ENREF_24}{}Hensch, T. K., 2005. Critical period plasticity in local cortical circuits. Nat Rev Neurosci 6, 877--88, \href{https://doi.org/10.1038/nrn1787}{doi:10.1038/nrn1787}.
\item \hypertarget{_ENREF_25}{}Hoekstra, R. H. A., de Ron, J., Epskamp, S., Robinaugh, D. J., Borsboom, D., 2026. Mapping the dynamics of idiographic network models to the network theory of psychopathology. Behav Res Methods 58, 130, \href{https://doi.org/10.3758/s13428-026-03009-w}{doi:10.3758/s13428-026-03009-w}.
\item \hypertarget{_ENREF_26}{}Kramers, H. A., 1940. Brownian motion in a field of force and the diffusion model of chemical reactions. Physica 7, 284--304, \href{https://doi.org/10.1016/S0031-8914(40)90098-2}{doi:10.1016/S0031-8914(40)90098-2}.
\item \hypertarget{_ENREF_27}{}Marsman, M., Borsboom, D., Kruis, J., Epskamp, S., van Bork, R., Waldorp, L. J., van der Maas, H. L. J., Maris, G., 2018. An Introduction to Network Psychometrics: Relating Ising Network Models to Item Response Theory Models. Multivariate Behav Res 53, 15--35, \href{https://doi.org/10.1080/00273171.2017.1379379}{doi:10.1080/00273171.2017.1379379}.
\item \hypertarget{_ENREF_28}{}Masuda, N., Islam, S., Aung, S. T., Watanabe, T., 2025. Energy landscape analysis based on the Ising model: Tutorial review. PLOS Complex Syst 2, e0000039, \href{https://doi.org/10.1371/journal.pcsy.0000039}{doi:10.1371/journal.pcsy.0000039}.
\item \hypertarget{_ENREF_29}{}Myin-Germeys, I., Kasanova, Z., Vaessen, T., Vachon, H., Kirtley, O., Viechtbauer, W., Reininghaus, U., 2018. Experience sampling methodology in mental health research: new insights and technical developments. World Psychiatry 17, 123--132, \href{https://doi.org/10.1002/wps.20513}{doi:10.1002/wps.20513}.
\item \hypertarget{_ENREF_30}{}Onuchic, J. N., Luthey-Schulten, Z., Wolynes, P. G., 1997. Theory of protein folding: the energy landscape perspective. Annu Rev Phys Chem 48, 545--600, \href{https://doi.org/10.1146/annurev.physchem.48.1.545}{doi:10.1146/annurev.physchem.48.1.545}.
\item \hypertarget{_ENREF_31}{}Seblani, M., Brezun, J. M., Feron, F., Hoquet, T., 2024. Rethinking plasticity: Analysing the concept of "destructive plasticity" in the light of neuroscience definitions. Eur J Neurosci 60, 4798--4812, \href{https://doi.org/10.1111/ejn.16487}{doi:10.1111/ejn.16487}.
\item \hypertarget{_ENREF_32}{}Tononi, G., Sporns, O., Edelman, G. M., 1999. Measures of degeneracy and redundancy in biological networks. Proc Natl Acad Sci U S A 96, 3257--62, \href{https://doi.org/10.1073/pnas.96.6.3257}{doi:10.1073/pnas.96.6.3257}.
\item \hypertarget{_ENREF_33}{}Toulouse, G., 1977. Theory of the frustration effect in spin glasses: I. Communications on Physics 2, 115--119.
\item \hypertarget{_ENREF_34}{}van der Maas, H. L. J., Dolan, C. V., Grasman, R. P., Wicherts, J. M., Huizenga, H. M., Raijmakers, M. E., 2006. A dynamical model of general intelligence: the positive manifold of intelligence by mutualism. Psychol Rev 113, 842--61, \href{https://doi.org/10.1037/0033-295X.113.4.842}{doi:10.1037/0033-295X.113.4.842}.
\item \hypertarget{_ENREF_35}{}Viglione, A., Delli Colli, C., Poggini, S., Mobasher, M., Vai, B., Poletti, S., Benedetti, F., Branchi, I., 2025. A network-based approach reveals higher plasticity levels in bipolar than major depressive disorder. J Affect Disord 387, 119549, \href{https://doi.org/10.1016/j.jad.2025.119549}{doi:10.1016/j.jad.2025.119549}.
\item \hypertarget{_ENREF_36}{}Waddington, C. H., 1957. The Strategy of the Genes: A Discussion of Some Aspects of Theoretical Biology. George Allen \& Unwin, London.
\item \hypertarget{_ENREF_37}{}Waldorp, L., Marsman, M., 2022. Relations between Networks, Regression, Partial Correlation, and the Latent Variable Model. Multivariate Behav Res 57, 994--1006, \href{https://doi.org/10.1080/00273171.2021.1938959}{doi:10.1080/00273171.2021.1938959}.
\item \hypertarget{_ENREF_38}{}Wales, D. J., 2003. Energy Landscapes: Applications to Clusters, Biomolecules and Glasses. Cambridge University Press, Cambridge.
\item \hypertarget{_ENREF_39}{}Watanabe, T., Hirose, S., Wada, H., Imai, Y., Machida, T., Shirouzu, I., Konishi, S., Miyashita, Y., Masuda, N., 2014. Energy landscapes of resting-state brain networks. Front Neuroinform 8, 12, \href{https://doi.org/10.3389/fninf.2014.00012}{doi:10.3389/fninf.2014.00012}.
\item \hypertarget{_ENREF_40}{}West-Eberhard, M. J., 2003. Developmental Plasticity and Evolution. Oxford University Press, New York.
\end{list}

\end{document}